%% file: main-2.tex
\documentclass[a4paper,UKenglish,cleveref, autoref, thm-restate]{lipics-v2019}
  \nolinenumbers
\usepackage[inline]{enumitem}
  
  \usepackage[utf8]{inputenc}
\usepackage[T1]{fontenc}
\usepackage{amsmath,amsthm,amssymb,amsfonts}
\usepackage{xspace}
\usepackage{hyperref}
\usepackage{cleveref}
\usepackage{verbatim}
\usepackage{tikz}
\usepackage{setspace}
\usetikzlibrary{snakes}
\usepackage[procnumbered,linesnumbered,ruled,vlined]{algorithm2e}
\hypersetup{
    colorlinks=true,
    citecolor=red,
    linkcolor=blue,
    filecolor=magenta,      
    urlcolor=black,
}
\usetikzlibrary{positioning, fadings, backgrounds}

\usepackage[textsize=footnotesize,color=green!40]{todonotes}

\usetikzlibrary{decorations.pathmorphing}
\tikzset{snake it/.style={decorate, decoration=snake}}

\newcommand{\IPC}{\textsc{Isometric Path Cover}\xspace}
\newcommand{\dist}[2]{d(#1,#2)}
\newcommand{\SP}[1]{{\cal S}(#1)}

\newcommand{\ipac}[1]{ipacc\left(#1\right)}

\definecolor{dartmouthgreen}{rgb}{0.05, 0.5, 0.06}

  \title{ Additive approximation algorithm for geodesic centers in $\delta$-hyperbolic graphs}

\author{Dibyayan Chakraborty}{School of Computing, University of Leeds, United Kingdom.}{}{}{}

\author{Yann Vax\`{e}s}{Laboratoire d'Informatique et Systèmes, Aix-Marseille Université and CNRS, Faculté des Sciences de Luminy, F-13288 Marseille, Cedex 9, France.}{}{}{}




\begin{document}

\maketitle

\setstretch{1}

\begin{abstract}

 % For a collection $\mathcal{C}$ of isometric paths and a vertex $v\in V(G)$, the distance between $v$ and $\mathcal{C}$ is the minimum of $\dist{v}{P}$ with $P\in \mathcal{C}$. 
% Is it useful to define the eccentricity of a family of paths before stating the problem ?

For an integer $k\geq 1$, the objective of \textsc{$k$-Geodesic Center} is to find a set $\mathcal{C}$ of $k$ isometric paths such that the maximum distance between any vertex $v$ and $\mathcal{C}$ is minimised. Introduced by Gromov, \emph{$\delta$-hyperbolicity} measures how treelike a graph is from a metric point of view. Our main contribution in this paper is to provide an additive $O(\delta)$-approximation algorithm for \textsc{$k$-Geodesic Center} on $\delta$-hyperbolic graphs. 
On the way, we define a coarse version of the pairing property introduced by Gerstel \& Zaks (Networks, 1994) and show it holds  for $\delta$-hyperbolic graphs. This result allows to reduce the \textsc{$k$-Geodesic Center} problem to its rooted counterpart, a main idea behind our algorithm. 
We also adapt a technique of Dragan \& Leitert, (TCS, 2017) to show that for every $k\geq 1$, $k$-\textsc{Geodesic Center} is NP-hard even on partial grids. 
\end{abstract}

\section{Introduction}

Given a graph $G,$ the \textsc{$k$-Geodesic Center} problem asks to find a collection $\cal C$ of $k$ isometric paths such that the maximum distance between any vertex and $\cal C$ is minimised. 
This problem may arise in determining a set of $k$ "most accessible" speedy line routes in a network and can find applications in communication networks, transportation planning, water resource management and fluid transportation \cite{dragan2017minimum}.
The decision version of this problem asks, given a graph $G$ and two integers $k$ and $R,$ whether there exist a collection $\cal C$ of $k$ isometric paths such that any vertex of $G$ is at distance at most $R$ from $\cal C.$ 

\textsc{$k$-Geodesic Center} is related to several algorithmic problems studied in the literature. \textsc{$k$-Geodesic Center} is a generalisation of \textsc{Minimum Eccentricity Shortest Path} (MESP)  where given an integer $R$, the objective is to decide if there exists an isometric path $P$ such that the maximum distance between any vertex and $P$ is at most $R$ \cite{dragan2017minimum}. Clearly, $1$-\textsc{Geodesic Center} is equivalent to \textsc{MESP}. If, instead of isometric paths, we asks whether there exists a subset of $k$ vertices of eccentricity at most $R,$ we obtain the decision version of \textsc{$k$-center} which is one of the most studied facility location problem in the literature~\cite{hochbaum1985best,hsu1979easy,lkacki2024fully, plesnik1980computational,TaFrLa}. The solution of a $k$-\textsc{Geodesic Center} can be thought of as a relaxation of \textsc{$k$-Center}. $k$-\textsc{Geodesic Center} is also related to \textsc{Isometric Path Cover}, where the objective is to find the minimum number of isometric paths that contains all vertices of the input graph. Study of the algorithmic aspects of \textsc{Isometric Path Cover} has garnered much attention recently~\cite{DBLP:conf/isaac/ChakrabortyD0FG22,dumas2022graphs,fernau2023parameterizing}.

% \textsc{$k$-Geodesic Center} generalizes several problems previously studied in the literature. When $R=0,$ we obtain the decision version of \IPC (IPC) asking whether there exist $k$ isometric paths covering all vertices of $G$ \cite{DBLP:conf/isaac/ChakrabortyD0FG22}. 
% When $k=1$, we obtain the decision version of the \textsc{Minimum Eccentricity Shortest Path} (MESP) problem asking whether there exists a shortest path of eccentricity at most $R$ \cite{dragan2017minimum}. 

All the three problems (i.e. \textsc{IPC, MESP}, and \textsc{$k$-Center}) are NP-hard for general graphs but are known to admit exact polynomial time algorithms when the given graph $G$ is a tree \cite{DBLP:conf/isaac/ChakrabortyD0FG22,dragan2017minimum,WaZh2021}. This raises the question about the complexity of these problems when the input graph is \emph{close to a tree}?
In this paper, we consider the graph parameter \emph{$\delta$-hyperbolicity}~\cite{gromov1987}, which measures how treelike a graph is from a metric point of view. See \Cref{sec:prelim} for a formal definition. Graphs with constant $\delta$-hyperbolicity are called \emph{hyperbolic} graphs. From a practical perspective, the study of $\delta$-hyperbolicity of graphs is motivated by the fact that many real-world graphs are tree-like \cite{AbuDra16,AdcSulMah13,JonLahBon08} or have small $\delta$-hyperbolicity \cite{BorCouCre+15,EdwardsKennedySanie2018,NaSa}. From a theoretical perspective, many popular graph classes like interval graphs, chordal graphs, $\alpha_i$-metric graphs \cite{DrDu-2023}, graphs with bounded tree-length~\cite{dourisboure2007tree}, link graphs of simple polygons ~\cite{chepoi2008diameters} have constant $\delta$-hyperbolicity.

Polynomial time approximation algorithms with an error (additive or multiplicative) depending only on the $\delta$-hyperbolicity of $G$ exist for \textsc{MESP}~\cite{mohammed2019slimness}, \textsc{$k$-Center} \cite{ChepoiEstellon07,EdwardsKennedySanie2018}, and \IPC \cite{chakraborty2023}. Motivated by the above results, in this paper, we provide an additive $O(\delta)$-approximation algorithm\footnote{A feasible solution for a minimization problem is said to be \emph{additive $\alpha$-approximate} if its objective value is at most the optimum value plus $\alpha$. An \emph{additive $\alpha$-approximation algorithm} for a minimization problem is a polynomial time algorithm that produces an additive $\alpha$-approximate solution for every instance of the input.} for \textsc{$k$-Geodesic Center} on $\delta$-hyperbolic graphs for arbitrary $k$. {The same algorithmic approach leads to an exact polynomial time algorithm in case of trees.}
% Namely, we prove the following two theorems:

% Connection to facility location problems, practical motivation.
% The problem of deciging whether there exists a collection $\mathcal C$ of $k$ geodesics such that the maximum distance between any vertex $v$ and  $\mathcal C$ is at most $R$ is a common generalization of the minimum eccentricity shortest path problem (Dragan and Leteirt, 2017) which is obtain for $k=1$, and the isometric path cover problem (Fisher and Fitzpatrick, 2001) which is obtaine for $R=0$. For k=1, connection to minimum distortion embedding on lines problems. 
%-> check the motivation of the papers that introduced these two problems
%-> introduce the notion of pairing, its history and explain why it is useful for our problem

\begin{theorem}
\label{thm:hyperbolic}
Let $G$ be a $\delta$-hyperbolic graph and $k$ be an integer. Then,
    there is a polynomial time $O(\delta)$-additive approximation algorithm for $k$-\textsc{Geodesic Center} on $G.$
\end{theorem}

Our algorithm has mainly two stages. In the first stage, we solve the ``rooted'' version of $(2k-1)$-\textsc{Geodesic Center}, where we require that all isometric paths in the solution has a common end-vertex. 
%To solve the $k$-\textsc{Geodesic Center}, we solve (with additive error depending only on $\delta$) the rooted version of the $(2k-1)$-\textsc{Geodesic Center}. 
Then to reduce the number of isometric paths, we use the \emph{shallow pairing} property of $\delta$-hyperbolic graphs. See \Cref{def:pairing}. Intuitively, this property ensures that the $2k$-many end-vertices of the $2k-1$ isometric paths obtained in the first stage can be ``paired'' to obtain $k$ many isometric paths which together provide an additive $O(\delta)$-approximation algorithm for \textsc{$k$-Geodesic Center}. We think that the shallow pairing property could also be interesting in itself and for other algorithmic applications.

We also adapt a technique of Dragan \& Leitert, (TCS '17) to show that for every $k\geq 1$, $k$-\textsc{Geodesic Center} is NP-hard even on \emph{partial grids}.
A graph is a \emph{partial grid} if it is a subgraph of $(k\times k)$-grid for some positive integer $k$.

\begin{theorem}\label{thm:hard}
    For every integer $k\geq 1$, $k$-\textsc{Geodesic Centre} is NP-hard even on partial grids. 
\end{theorem}

\medskip\noindent\textbf{Related Works} To the best of our knowledge, the computational complexity of \textsc{$k$-Geodesic Center} for arbitrary $k$ has not been studied before. Therefore, we begin by surveying the relevant results on $1$-\textsc{Geodesic Center} i.e. the MESP problem. Dragan \& Leitart~\cite{dragan2017minimum} gave several constant factor approximation algorithms for MESP with varying running times on general graphs. In an another paper~\cite{JGAA-394}, the authors proposed polynomial time algorithms for MESP on graph classes like \emph{chordal} graphs and \emph{distance hereditary} graphs. In fact the authors proved that, MESP admits an $O(n^{\gamma+3})$-time algorithm on graphs with \emph{projection gap} at most $\gamma$, and $n$ vertices. The parameter \emph{projection gap} generalises the notion of $\delta$-hyperbolicity. Their result implies that MESP admits an $O(n^{4\delta+4})$-time algorithm on graphs with $\delta$-hyperbolicity at most $\delta$, and $n$ vertices. We do not know if MESP admits a fixed parameter algorithm with respect to $\delta$-hyperbolicity. The same authors also proposed additive approximation algorithms for graphs with bounded \emph{tree-length}. Fixed parameter tractability of MESP with respect to various graph parameters like modular width, distance to cluster
graph, maximum leaf
number, feedback edge set, etc. have been also studied recently~\cite{bhyravarapu2023parameterized,kuvcera2023minimum}. As noted by Ku{\v{c}}era and Such{\`y}~\cite{kuvcera2023minimum}, the fixed parameter tractablility of MESP with respect to \emph{tree-width} is an interesting open problem. (Tree-width measures how far a graph is from a tree from a structural point of view.) Relation of MESP with other problems like the \emph{minimum distortion embedding on a line}~\cite{dragan2017minimum} and \emph{$k$-laminar} problem~\cite{birmele2016minimum} have been established.

Dragan \& Leitart~\cite{JGAA-394} observed that MESP admits an additive $O(\delta \log n)$-approximation algorithm on graphs with $\delta$-hyperbolicity at most $\delta$ and $n$ vertices. Their proof uses the fact that the tree-length of $\delta$-hyperbolic graphs are at most $O(\delta \log n)$. As the best known bound for tree-length of $\delta$-hyperbolic graphs is  $O(\delta \log n)$, this method seems not directly provide constant error in case of hyperbolic graph. 
% because the best bound we have for the tree-length of $\delta$-hyperbolic graphs is $O(\delta \log n)$ Since there exists graphs whose tree-length is $\Omega(\delta \log n)$ \todo{ 7-systolic graphs (i.e. plane triangulation with inner vertices of degrees at least 7)are 1-hyperbolic and I think there is no tree-decomposition with bags of diameter $O(\log n)$ but I have no citation in mind. Maybe we should rephrase and say that the method cannot be directly used to treat  ?}, their technique cannot be directly generalised. 
Then, in the PhD Thesis of A.O. Mohammed~\cite{mohammed2019slimness}, an $O(\delta)$-approximation algorithm for MESP on $\delta$-hyperbolic graphs have been proposed. Other examples include fast additive $O(\delta)$-approximation algorithms for finding the diameter, radius, and all eccentricities~\cite{chepoi2008diameters,chepoi2019fast,ChepoiEstellon07} as well as packing and covering for families of quasiconvex sets \cite{chepoi2017core}. \Cref{thm:hyperbolic} adds \textsc{$k$-Geodesic Center} in the list of problems admitting an additive approximation algorithm depending only on the $\delta$-hyperbolicity of the input graph. Recently, the computational complexity of maximum independent set of planar $\delta$-hyperbolic graphs have been studied~\cite{kisfaludi2023separator}.

\medskip\noindent\textbf{Organisation:} In \Cref{sec:prelim} we introduce some terminologies. In \Cref{sec:pairing}, we introduce the notion of shallow pairing and prove its existence in $\delta$-hyperbolic graphs. In \Cref{sec:algorithm,sec:hard}, we prove \Cref{thm:hyperbolic,thm:hard}, respectively.

% Connection to facility location problems, practical motivation.
% The problem of deciging whether there exists a collection $\mathcal C$ of $k$ geodesics such that the maximum distance between any vertex $v$ and  $\mathcal C$ is at most $R$ is a common generalization of the minimum eccentricity shortest path problem (Dragan and Leteirt, 2017) which is obtain for $k=1$, and the isometric path cover problem (Fisher and Fitzpatrick, 2001) which is obtaine for $R=0$. For k=1, connection to minimum distortion embedding on lines problems. 
%-> check the motivation of the papers that introduced these two problems
%-> introduce the notion of pairing, its history and explain why it is useful for our problem

\newcommand{\geodtriangle}[3]{\Delta\left(#1,#2,#3\right)}
\newcommand{\isom}[2]{\sigma(#1,#2)}
\newcommand{\grom}[3]{\left(#1|#2\right)_{#3}}
\newcommand{\ballp}[2]{B_{#1}\left({#2}\right)}
\newcommand{\opt}[1]{R^*_{#1}}

\section{Preliminaries}\label{sec:prelim}

\noindent \textbf{Basic notations:} For two vertices $u,v\in V(G)$, $\isom{u}{v}$ shall denote an $(u,v)$-isometric path in $G$ and the length (i.e. the number of edges) in $\isom{u}{v}$ is denoted as $\dist{u}{v}$, the \emph{distance} between $u$ and $v$. If an isometric path $P$ of $G$ has a vertex $r$ as end vertex, then $P$ is called an \emph{$r$-path}.

For two sets $S_1,S_2 \subseteq V(G)$ of vertices, $\dist{S_1}{S_2} = \min\{ \dist{u}{v} \colon u\in S_1, v\in S_2\}$ is the distance between $S_1$ and $S_2$. For convenience, if one subset of vertices is a singleton, we abbreviate $d(\{v\},S)$ by $d(v,S).$ 
For an integer $k$ and a set of vertices $S$, the \emph{$k$-neighborhood (or $k$-ball) around $S$}, denoted as $\ballp{k}{S}$, is the set of all vertices $v$ such that $\dist{v}{S}\leq k$. For an integer $R$, a collection $\mathcal{C}$ of isometric paths of $G$ is an \emph{$R$-cover} of $G$ if $$\displaystyle \bigcup\limits_{P\in \mathcal{C}} \ballp{R}{V(P)} = V(G)$$  For an integer $k$, the symbol $\opt{k}$ shall denote the minimum integer for which there is a $\opt{k}$-cover $\mathcal{C}$ of $G$ with $|\mathcal{C}|=k$. If every path in $\mathcal{C}$ is an $r$-path, then $\mathcal{C}$ is an \emph{$(r,R)$-cover} of $G$.  
For an integer $R$, a vertex $r$ and a subset $S\subseteq V(G)$ is a \emph{$(r,R)$-packing} if the $R$-neighborhood of any $r$-path $P$ contains at most one vertex of $S,$ i.e. $|\ballp{R}{V(P)} \cap S|\le 1$. Note that if $S$ is an $(r,R)$-packing of $G$ then any $(r,R)$-cover of $G$ has size at least $|S|.$ Indeed, the $R$-neighborhood of any $r$-path covers at most one vertex in $S.$ Hence, it is not possible to cover $S$ with less than $|S|$ $r$-paths.

% \subsection{Notations related to $\delta$-hyperbolicity} 
%\subsection{Hyperbolicity}\label{sec:hyperbolicity}
\medskip \noindent \textbf{Definitions related to $\delta$-hyperbolicity}:
Let $(X,d)$ be a metric space. A geodesic segment joining two points $x$ and $y$ from $X$ is a map $\rho$ from the segments $\isom{a}{b}$ of length $|a-b| = \dist{x}{y}$ to $X$ such that $\rho(a) = x,$ $\rho(b)=y$, and $\dist{\rho(s)}{\rho(t)}=|s-t|$ for all $s,t\in \isom{a}{b}.$ A metric space $(X,d)$ is geodesic if every pair of points in $X$ can be joined by a geodesic. We will denote by $\isom{x}{y}$ any geodesic segment connecting the points $x$ and $y.$ 

Introduced by Gromov \cite{gromov1987}, $\delta$-hyperbolicity measures how treelike a graph is from a metric point of view. Recall that a metric space $(X,d)$ embeds into a tree network (with positive real edge lengths), if and only if for any four points $u,v,w,x$ the two larger of the distance sums $\dist{u}{v}+\dist{w}{x},$ $\dist{u}{w}+\dist{v}{x},$ and $\dist{u}{x}+\dist{v}{w}$ are equal. A metric space $(X,d)$ is called $\delta$-hyperbolic if the two larger distance sums differ by at most $2\delta.$ 
For a metric space $(X,d),$ the \emph{Gromov product} of two points $x,y$ with respect to a third point $z$ is defined as $$\grom{x}{y}{z} = \frac{1}{2}\left( \dist{x}{z} + \dist{z}{y} - \dist{x}{y} \right)$$ 
Equivalently, a metric space $(X,d)$ is $\delta$-hyperbolic if for any four points $u,v,w,x,$ $$\grom{u}{w}{x} \ge \min\left\{\grom{u}{v}{x},\grom{v}{w}{x}\right\}-\delta.$$
A connected graph $G=(V,E)$ equipped with standard graph metric $d_G$ is $\delta$-hyperbolic if $(V,d_G)$ is a $\delta$-hyperbolic metric space. The $\delta$-hyperbolicity $\delta(G)$ of a graph $G$ is the smallest $\delta$ such that $G$ is $\delta$-hyperbolic.

There exists several equivalent definitions of $\delta$-hyperbolic metric spaces involving different but comparable values of $\delta.$ In the proof of \Cref{thm:hyperbolic}, we will use the definition employing $\delta$-thin geodesic triangles. A {\em geodesic triangle} $\geodtriangle{x}{y}{z}$ is a union $\isom{x}{y}\cup \isom{x}{z}\cup \isom{y}{z}$ of three geodesic segments connecting these vertices. Let $m_x$ be the point of the geodesic $\isom{y}{z}$ located at distance $\alpha_y:=\grom{x}{z}{y}$ from $y.$ Then $m_x$ is located at distance $\alpha_z:=\grom{x}{y}{z}$ from $z$ because $\alpha_y+\alpha_z=\dist{y}{z}.$ Analogously, define the points $m_y\in \isom{x}{z}$ and $m_z\in \isom{x}{y}$ both located at distance $\alpha_x:=\grom{y}{z}{x}$ from $x$; see \Cref{fig:thin-triangle}. 

\begin{figure}
    \centering
    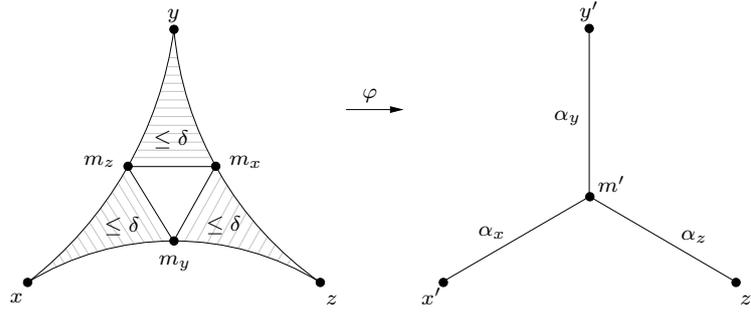
    \caption{A $\delta$-thin geodesic triangle}
    \label{fig:thin-triangle}
\end{figure}

There exists a unique isometry $\varphi$ which maps the geodesic triangle $\geodtriangle{x}{y}{z}$ to a star $T(x,y,z)$ consisting of three solid segments $\isom{x'}{m'},$ $\isom{y'}{m'},$ and $\isom{z'}{m'}$ of length $\alpha_x,$ $\alpha_y,$ and $\alpha_z,$ respectively. This isometry maps the vertices $x,y,z$ of $\geodtriangle{x}{y}{z}$ to the respective leaves $x',y',z'$ of $T(x',y',z')$ and the points $m_x,$ $m_y,$ and $m_z$ to the center $m$ of this tripod. Any other point of $T(x',y',z')$ is the image of exactly two points of $\geodtriangle{x}{y}{z}.$ A geodesic triangle $\geodtriangle{x}{y}{z}$ is called $\delta$-thin \cite{AlBrCoFeLuMiShSh} if for all points $u,v \in \geodtriangle{x}{y}{z},$ $\varphi(u)=\varphi(v)$ implies $\dist{u}{v}\le \delta.$ 

The following result shows that the $\delta$-hyperbolicity of geodesic space can be approximated by the maximum thinness of its geodesic triangles. 

\begin{proposition}[\cite{AlBrCoFeLuMiShSh,BrHa,gromov1987}]\label{prop:hyp-thin}
    Geodesic triangles of geodesic $\delta$-hyperbolic space are $4\delta$-thin. Conversely, geodesic space with $\delta$-thin triangles are $\delta$-hyperbolic.
\end{proposition}

Every graph $G=(V,E)$ equipped with its standard distance $d_G$ can be transformed into a geodesic space $(X,d)$ by replacing every edge $e=(u,v)$ by a solid segment $[u,v]$ of length $1.$ These segments may intersect only at their common ends. Then $(V,d_G)$ isometrically embeds naturally in $(X,d).$ 
The \emph{thinness} $\tau(G)$ of a graph $G$ is the smallest integer $\delta$, such that all geodesic triangles of the geodesic space arising from $G$ are $\delta$-thin. When thinness of a graph is $\delta$, then it is also called a \emph{$\delta$-thin} graph.

\medskip \noindent \textbf{Subdivisions and partial grids}: For a graph $G$, its \emph{$\ell$-subdivision}, denoted as $G_{\ell}$, is obtained by replacement of all its edges by paths of a fixed length $\ell\geq 1$. An \emph{equal subdivision} of $G$ is an $\ell$-subdivision for some $\ell\geq 1$. The vertices of $G$ in $G_{\ell}$ are the \emph{original vertices}. We shall use the following result.

\begin{lemma}[\cite{chakraborty2023}]\label{lem:equal}
    Let $G$ be a planar graph with maximum degree $4$. Then there exists a partial grid graph $H$, which is an equal subdivision of $G$ and contains at most $O(|V(G)|^3)$ vertices.
\end{lemma}

\section{Pairings and shallow pairings} \label{sec:pairing} In this section we first recall the definition of {\em the pairing property} following the terminology of \cite{bénéteau2022abctgraphs}. Then, we introduce a coarse version of this property, called {\em shallow pairing property}\footnote{The notion of approximate (shallow) pairing in hyperbolic graphs was
defined by Victor Chepoi, who also asked the question about their
existence.}, and show that this relaxed property holds for $\delta$-hyperbolic graphs. 

Given a connected graph $G$, a {\it profile} is any finite sequence $\pi = (x_1, \ldots, x_n)$ of
vertices of $G$.  The \emph{total distance} of a vertex $v$ of $G$ is defined by $F_{\pi}(v)=\sum_{i=1}^n d(v,x_i).$ A \emph{pairing} $P$ is a partition of an even profile $\pi$ of length $k=2n,$ into $n$ disjoint pairs. For a pairing $P$, define $D_{\pi}(P)=\sum_{\{ a,b\}\in P} d(a,b)$.  
%A pairing $P$ of $\pi$ maximizing the function $D_{\pi}$ is called a \emph{maximum pairing}.  
The notion of pairing was defined by Gerstel and Zaks~\cite{GeZa}; they also proved the following weak duality between the functions $F_{\pi}$ and $D_{\pi}$:

\begin{lemma}[\cite{GeZa}]\label{pairing1}
  For any even profile $\pi$ of length $k=2n$ of a connected graph $G$, for any pairing
  $P$ of $\pi$ and any vertex $v$ of $G$, $D_{\pi}(P)\le F_{\pi}(v)$
  and the equality holds if and only if
  $v \in \bigcap_{\{a,b\} \in P} I(a,b)$.
\end{lemma}

We say that a graph $G$ satisfies the \emph{pairing property} if for any even profile $\pi$ there exists a pairing $P$ of $\pi$ and a vertex $v$ of $G$ such that $D_{\pi}(P)=F_{\pi}(v)$, i.e., the
functions $F_{\pi}$ and $D_{\pi}$ satisfy the strong duality. Such a pairing is called a \emph{perfect pairing}. By \Cref{pairing1}, the pairing property of~\cite{GeZa} coincides with the \emph{intersecting-intervals property} of~\cite{McMoMuNoPo}. It was shown
in~\cite{GeZa} that trees satisfy the pairing property. More generally, it was shown in~\cite{McMoMuRo} and independently
in~\cite{ChFeGoVa} that cube-free median graphs also satisfy the
pairing property. It was proven in~\cite{McMoMuNoPo} that the complete bipartite graph $K_{2,n}$ satisfies the pairing property. As observed
in~\cite{McMoMuNoPo}, the $3$-cube is a simple example of a graph not satisfying the pairing property. 
%The investigation of the structure of graphs with the pairing property was formulated as an open problem in~\cite{GeZa}.

In general, $\delta$-hyperbolic graphs do not satisfy the pairing property but, as shown below, they satisfy some coarse variant of the pairing property. Before defining this variant, let us first define the notion of {\it $\gamma$-shallow pairing}
\begin{definition}[$\gamma$-shallow pairing]\label{def:pairing}
    Let $G$ be a graph and $\pi$ an even profile of length  $2k$. A \emph{$\gamma$-shallow pairing} of $\pi$ is a pairing $P$ such that, there exists a vertex $v$ with $\grom{x}{y}{v}\le \gamma$ for every $\{x,y\}\in P.$
%is a partition ${\cal S}$ of $X$ into $k$ disjoint pairs such that there exists a vertex $m\in V(G)$ with $(x|y)_m \leq 2\delta+\frac{1}{2}$ for every pair $\{x,y\} \in {\cal P}$. 
%     \begin{enumerate*}[label=(\textbf{\Alph*})]
%         \item end vertices of all paths in $\mathcal{C}$ lie in $X$, and no two paths in $\mathcal{C}$ have a common end vertex;
%         \item there exists a vertex $m\in V(G)$ such that $(x|y)_m \leq 2\delta+1$ for the end-vertices $x$ and $y$ of any path $P\in \mathcal{C}.$
%    \end{enumerate*}
\end{definition}
 In the definition of a perfect pairing $P$, the vertex $v$ belongs to every interval between a pair of vertices in $P.$ As the following lemma shows, in the definition of a $\gamma$-shallow pairing $P_\gamma$, the vertex $v$ is at distance at most $\gamma+\tau(G)$ from every isometric path between a pair of vertices in $P_\gamma.$
 \begin{lemma}
 Let $G$ be a graph, $\pi$ be an even profile of length $2k$ and $P$ be a $\gamma$-shallow pairing of $\pi$. Then there exists a vertex $v$ such that $d(v,\isom{x}{y})\le \gamma+\tau(G)$ for every geodesic $\isom{x}{y}$ with $\{x,y\}\in P.$ 
 \end{lemma}
 \begin{proof}
 By definition of a $\gamma$-shallow pairing, there exists a vertex $v$ such that $\grom{x}{y}{v}\le \gamma$ for every $\{x,y\}\in P.$ For any pair of vertices $\{x,y\}\in P,$ consider the geodesic triangle $\Delta(x,y,v)$ and let $\isom{v}{x},\isom{v}{y},$ and $Q$ be the sides of this triangle. Let $x',y'$ be the points of $\isom{v}{x}$ and $\isom{v}{y}$, respectively,  located at distance $\grom{x}{y}{v}$ from $v$. Since the thinness of $G$ is $\tau(G),$ $\dist{x'}{y'}\le \tau(G)$, moreover $\dist{x'}{z'}\le \tau(G)$ and $\dist{y'}{z'}\le \tau(G)$, where $z'$ is the point of $Q$ at distance $\grom{y}{v}{x}$ from $v$ and at distance $\grom{x}{v}{y}$ from $y$. Since, $\grom{x}{y}{v}\le \gamma$, we conclude that $d(v,Q)\le d(v,z')\le d(v,x')+d(x',z')\le \gamma+\tau(G)$. %This establishes \Cref{lem:algo-shallow-pairing}.
 \end{proof}
 We say that a graph $G$ satisfies the \emph{$\gamma$-shallow pairing property} if for any even profile $\pi$ there exists a $\gamma$-shallow pairing $P$ of $\pi$.

The existence and the computation in polynomial time of a $(2\delta+\frac{1}{2})$-shallow pairing in a $\delta$-thin graph $G$ can be obtained using the concept of {\em fiber} that was introduced in \cite{chepoi2017core}.
For a vertex $u\in V(G)$ and a profile $\pi$, the {\it fiber} of $x$ with respect to a vertex $u$ is the set of vertices $$F_{u}(x)=\{ y\in \pi: (x|y)_{u}\ge 2\tau(G)+1\}.$$ From Claim 1 and 2 of \cite{chepoi2017core}, the following lemma holds.
\begin{lemma}\label{lem:fiber-size}
    For any graph $G$ and any even profile $\pi$ of length $2k,$ there is a vertex $v\in V(G)$ such that $|F_{v}(x)|\le k$ for any vertex $x\in \pi.$ 
\end{lemma}
Lemma \ref{lem:fiber-size} is useful to prove the following result:
\begin{proposition}\label{lem:algo-shallow-pairing}
    Any graph $G$ satisfies the \emph{$(2\tau(G)+\frac{1}{2})$-shallow pairing property}. Moreover, for any graph $G$ with $n$ vertices and $m$ edges, and any even profile $\pi$ of length $2k,$ a $(2\tau(G)+\frac{1}{2})$-shallow pairing of $\pi$ can be computed in $O(mn^2)$ time. 
\end{proposition}
\begin{proof}
First, we will prove that the vertex $v$ whose existence is guaranteed by Lemma \ref{lem:fiber-size} can be calculated efficiently.    %\todo{Can we state the two claims to make the paper self-contained? {\bf Answer:} I am not sure because the statement of Claim 1 in \cite{chepoi2017core} refer to a vertex $m^*$ which is the median of the profile $\pi$ in the injective hull (or Hellification) of G and I don't want to define the injective hull. Then, we should argue that any vertex of the injective hull is at distance at most $\delta$ from some vertex $m$ of $G*$. Finally, Claim 2 says that if we define the fibers with to respect to $m$ then we get what we need. A solution could be to add a claim saying that the number of vertices of $\pi$ that belong to any fiber is at most $k$ and to refer to Claim 1 and 2 of \cite{chepoi2017core} for the proof of this claim. What do you think ? Yes, this should be enough, I think.} 
%of \cite{chepoi2017core}, 
 Indeed, an $O(mn^2)$ time algorithm was given in \cite{ChChDrDu+} to compute the thinness $\tau(G)$ of a graph $G$ with $n$ vertices and $m$ edges. The matrix of distances between every pair of vertices can be computed in $O(mn)$ time. Then, for every vertex $u\in V(G)$, it is possible to compute in $O(k^2)$ time the Gromov products $\grom{x}{y}{u}$ for every pair of vertices $x,y\in \pi$. Within the same running time, it is possible to also compute the fibers $F_{u}(x)$ for every $x\in \pi.$ Indeed, it suffices to add $y$ to $F_u(x)$ and $x$ to $F_u(y)$ when we compute a value $\grom{x}{y}{u}$ exceeding $ 2\tau(G)+\frac{1}{2}.$ If the size of a fiber $F_u(x)$ becomes larger than $k$ then we can abort and try the next vertex $u.$ The procedure stops once we have find a vertex $v$ such that $F_v(x)\le k$ for every vertex $x\in \pi.$ By Lemma \ref{lem:fiber-size}, this happens for at least one vertex $v\in V(G).$ 
Hence, it is possible to find in $O(mn+nk^2)$ time a vertex $v$ such that $|F_{v}(x)| \le k$ for every $x\in \pi.$ Let $H$ be the graph defined on the vertices of $\pi$ by adding an edge between two vertices $x$ and $y$ whenever $\grom{x}{y}{v}\le 2\tau(G)+\frac{1}{2}.$ By the choice of $v,$ every vertex of $H$ has degree at least $k.$ By Dirac's theorem, $H$ is Hamiltonian and thus has a perfect matching $M.$ Such a perfect matching can be computed in $O(\sqrt{n}m)$ time \cite{blum1990new, micali1980v}. The pairing defined by the end-vertices of edges in $M$ is a $(2\delta+\frac{1}{2})$-shallow pairing of $\pi$ that can be computed in $O(mn^2)$ time. Indeed, the running time is dominated by the algorithm that computes the thinness of $G$. 
    %\todo[inline]{Todo}
\end{proof}

% We shall use the following lemma.

% For a positive integer $k$, a set $\mathcal{C}$ of isometric paths is ``$k$-\emph{shallow}'' if $|\mathcal{C}|=k$ and there exists a subset $X \subseteq V(G)$ with $|X| = 2k$ such that $\mathcal{C}$ is a shallow pairing of $X$.

% \section{Important lemma and observations}

In the next section, for any pairing $P,$ we will denote by $\SP{P}$ a collection of isometric paths between every pair of vertices in $P,$ i.e. $\SP{P}:=\left\{\isom{x}{y} : \{x,y\}\in P\right\}$

\section{Additive approximation algorithm}\label{sec:algorithm}

In this section, we prove \Cref{thm:thinness} which implies \Cref{thm:hyperbolic} by \Cref{prop:hyp-thin}.

\begin{theorem}\label{thm:thinness}
    Let $G$ be a graph with $m$ edges, $n$ vertices and $k$ be an integer. Then, Algorithm \ref{algo:k-geodesic center} is a $O(mn^2\log n)$-time $(6\tau(G)+1)$-additive approximation algorithm for $k$-\textsc{Geodesic Center} on $G$.
\end{theorem}

We provide a brief outline of the proof for the above theorem and organisation of this section.
%\todo{We may move the outline of the proof at the beginning of Section 3. Indeed, to be correct the proof should refer to thinness (introduced later) rather than hyperbolicity. Anyway, this outline seems to me useful to give to the reader an intuition of how the proof proceeds. It might also hepl to structure the description of the algorithm. 3.1. solution of size $k$ and eccentricity $R$ $\Rightarrow$ rooted solution of size $2k-1$ and eccentricity $R+\delta$ 3.2. rooted solution of size $2k-1$ and eccentricity $R$ $\Rightarrow$ solution of size $k$ and eccentricity $R+3\delta+1$  3.3. Primal dual algorithm.
%We might also consider changing the order of these three subsections.} 
A collection $\mathcal{C}$ of isometric paths is ``rooted'' if all the paths in $\mathcal{C}$ have a common end-vertex. 
First we show in \Cref{sec:rooted}, that the rooted version of the \textsc{$k$-Geodesic Center} problem 
%\todo{Since the rooted solution constructed had $2k$-paths and we (probably) compared the quality with the optimum of $k$-\textsc{Geodesic Center}, can we say that solved the rooted vertion of $k$-\textsc{Geodesic Center}?{\bf Answer}: When $R$ is the minimum integer for which Algorithm 1 returns a cover. We know that it returns a packing of size $2k$ for $R-1.$ Hence, there is no collection of $2k-1$ isometric rooted paths of eccentricity $R-1.$ Since the rooted cover returned by Algorithm 1 for $R$ has eccentricity at most $R+2\delta$, it computes in polynomial a rooted solution which is optimal up to an $O(\delta)$ error.} 
where we require that the collection of isometric paths is rooted can be solved in polynomial time up to an additive $2\delta$ error in $\delta$-thin graphs. For that, we use a primal dual algorithm and a dichotomy to find an integer $R$ such that there is a collection of $2k-1$ isometric rooted paths of eccentricity $R+2\delta$ and no such collection has eccentricity smaller than $R.$ 
%Notice that the idea of solving a rooted version of the problem was already used in \cite{chakraborty2023} to provide an $O(\delta)$-approximation algorithm for the isometric path cover problem. 
Then in \Cref{sec:non-rooted}, we show that any collection $\cal C$ of $k$ isometric paths can be transformed into a rooted collection $\cal C'$ of size $2k-1$ such that the eccentricity of $\cal C'$ is at most the eccentricity of $\cal C$ plus $\delta$. From this observation and the choice of $R$, we conclude that no collection of $k$ isometric paths has eccentricity smaller than $R-\delta.$ To transform the rooted collection returned by the primal-dual algorithm into a non rooted collection of size $k,$ we also need a converse result. For that, using the $(2\delta+\frac{1}{2})$-shallow pairing property of $\delta$-thin graphs, in \Cref{sec:shallow-pairing-use}, we show that any rooted collection $\cal C'$ of $2k-1$ isometric paths can be transformed into a non rooted collection $\cal C$ of $k$ isometric paths such that the eccentricity of $\cal C$ is at most the eccentricity of $\cal C'$ plus $3\delta+1.$ We complete the proof in \Cref{sec:complete} as follows: the rooted collection of eccentricity $R+2\delta$ computed by our primal-dual algorithm can be transformed into a collection of size $k$ and eccentricity $R+5\delta+1.$ Since there is no such collection with an eccentricity less than $R-\delta,$ the collection of eccentricity $R+5\delta+1$ is optimal up to a $6\delta+1$ error.

\subsection{An algorithm for the \textsc{Rooted $k$-Geodesic Center} problem } \label{sec:rooted}

In this Section, we present an algorithm that, given a $\delta$-thin graph $G,$  computes an integer $R$ such that no collection of $2k-1$ rooted isometric paths has eccentricity smaller than $R$ and there is a collection of $2k-1$ isometric rooted paths of eccentricity $R+2\delta.$ 
Our description of this algorithm proceeds in two steps. First, we describe Algorithm \ref{algo:rooted-paths}. Given a graph $G,$ a root $r \in V(G)$ and integer $R,$ Algorithm \ref{algo:rooted-paths} outputs either an $(r,R+2\delta)$-cover of $G$ of size $2k-1$ or a $(r,R)$-packing of size $2k$. Then, Algorithm \ref{algo:guaranteed-rooted-path} uses Algorithm \ref{algo:rooted-paths} to perform a dichotomy. For every vertex $u\in V(G),$ Algorithm \ref{algo:guaranteed-rooted-path} computes the smallest value $R_u$ for which Algorithm \ref{algo:rooted-paths} outputs a cover. We show that $R:=\min\{R_u : u\in V(G)\}$ is an integer such that no collection of $2k-1$ rooted isometric paths has eccentricity smaller than $R$ and there is a collection of $2k-1$ isometric rooted paths of eccentricity $R+2\delta.$ We start with the following technical lemma.

\begin{lemma}
    \label{GP-Bound}
    Let $\isom{u}{v}\cup \isom{v}{w} \cup \isom{u}{w}$ be a geodesic triangle with $\dist{w}{\isom{u}{v}}\le R.$ Then, $(u|v)_w\le R.$  
\end{lemma}

\begin{proof}
    Let $y$ be a vertex of $\isom{u}{v}$ at distance at most $R$ from $w.$ Let $y'$ be the vertex of $\Delta(u,v,w)$ such that $\varphi(y)=\varphi(y').$ Without loss of generality, we can assume that $y'\in \isom{v}{w}.$ By triangle inequality, $\dist{w}{v}=\dist{w}{y'}+\dist{y'}{v}\le \dist{w}{y}+\dist{y}{v}.$ Since $\dist{v}{y'}=\dist{v}{y},$ we get $\dist{w}{y'}\le \dist{w}{y}.$ Hence, $(u|v)_w\le \dist{w}{y'}\le \dist{w}{y} \le R.$ 
\end{proof}

\begin{algorithm}[t]
    \DontPrintSemicolon
   \begin{small}
   \SetKwInOut{KwIn}{Input}
       \SetKwInOut{KwOut}{Output}
       \KwIn{A $\delta$-thin graph $G$, $r\in V(G)$, $R\in \mathbb{N}$, an integer $k\leq |V(G)|$}
       \KwOut{$(r,R+2\delta)$-cover of $G$ of size $2k-1$ or an $(r,R)$-packing of size $2k$.}
       
       $X=V(G)$; $i=0$;
       
       \While{ $i\leq 2k-1$ and $X\neq \emptyset$}{
           Let $v_i\in X$ be a vertex with $\dist{r}{v_i} \geq \dist{r}{z}$ for all $z\in X$;
   
           Let $\sigma_i$ be any $(r,v_i)$-isometric path;
   
           %Let $A_i=\{x \in X\colon v_i \in \specialset{r}{x}{R}{X}\}$; 
           
           %$X_i= \displaystyle\bigcup\limits_{x\in A_i} \specialset{r}{x}{R}{X}$;
   
           $X_i=\left\{u\in X:\ \exists\ \text{an $r$-path}\ P\ \text{such that}\ d(u,P)\le R\  \text{and}\ d(v_i,P)\le R \right\}$
   
           $X = X\setminus X_i$;
   
           $\mathcal{P} = \mathcal{P}\cup \{v_i\}$;
           
           $\mathcal{C}=\mathcal{C} \cup \{ \sigma_i\}$.
   
           $i = i+1$;
       }
       
      \If {$X=\emptyset$}
      {\Return $\mathcal{C}$;}
      \Else {\Return $\mathcal{P}$ }

      \caption{}\label{algo:rooted-paths}
       
   \end{small}
    \end{algorithm}

   %For technical reasons, we introduce the following notation. For two vertices $r,x\in V(G)$, a subset $X$ of $V(G)$ and an integer $k$, let $\specialset{r}{x}{k}{X}$ denotes the set $\left\{ y \colon y \in X, \exists P=\isom{r}{x}, y\in \ballp{k}{P}\right\}$. 
   % $\specialset{r}{x}{k}{X} := \ballp{k}{I(r,x)}\cap X.$
   %In other words, the set $\specialset{r}{x}{k}{X}$ contains all those vertices in $X$ that lie in a $k$-neighborhood of some $(r,x)$-isometric path. In the following proof, we shall use the notations defined in Algorithm~\ref{algo:rooted-paths}.

   \begin{lemma}
       \label{lem:rooted-paths}
       %Let $G$ be a $\delta$-thin graph, $r$ be a vertex of $G,$ and $R,k$ be two integers. Then, 
       Algorithm~\ref{algo:rooted-paths} either returns an $(r, R)$-packing of size $2k$ or an $(r,R+2\delta)$-cover of $G$ of size at most $2k-1$.
   \end{lemma}
   
   \begin{proof}
   First assume that Algorithm~\ref{algo:rooted-paths} returns a subset of vertices $\mathcal{P}$. Suppose there exists two vertices $\{v_i,v_j\}\subseteq \mathcal{P}$ with $i<j$ such that 
   \begin{enumerate*}[label=(\textbf{\alph*})]
       \item $v_i,v_j$ were included in $\mathcal{P}$ at the $i^{th}$ and $j^{th}$ iteration of Algorithm~\ref{algo:rooted-paths}, and 
       \item there exists an $r$-path $P$ such that $\{v_i,v_j\} \subseteq \ballp{R}{V(P)} \cap \mathcal{P}$. But then $v_j\in X_i$ and therefore was removed from $X$ in the $i^{th}$ iteration, a contradiction.
   \end{enumerate*} 
   % Observe that, the vertex $v_j$ does not lie in $X_i $ which implies that the distance between any isometric path $\isom{r}{v_j}$ and $v_i$ is greater than $R$. 
   
   Now assume that Algorithm~\ref{algo:rooted-paths} returns a collection of $r$-paths $\mathcal{C}$. For a vertex $u\in V(G)$, let $v_i\in \mathcal{P}$ be the vertex such that $u\in X_i$ when $u$ was removed from $X.$ By definition of $X_i,$ there exists an $r$-path $P$ such that $\dist{u}{P}\le R$ and $\dist{v_i}{P}\le R.$
   %$\{u,v_i\}\subseteq \ballp{R}{V(P)}$. 
   Let $x$ be the end-vertex of $P$ distinct from $r.$ Let $\sigma_i$ be the $r$-path added to $\cal C$ by Algorithm~\ref{algo:rooted-paths} during the $i^{th}$ iteration and $\isom{r}{u}$ be any isometric path between $r$ and $u$. We distinguish two cases (see \Cref{fig:Lemma-rooted-paths}).  
   \begin{itemize}
       \item {\bf Case 1.} First suppose that $(u|x)_r \le (v_i|x)_r.$ Let $m_i,$ $m'_i,$ $m''_i$ be the points of $\sigma_i,$ $P$ and $\isom{r}{u}$ at distance $(u|x)_r$ from $r.$ Since $\dist{u}{\isom{r}{x}}\le R,$ \Cref{GP-Bound} implies $\dist{u}{m''_i}=(r|x)_u\le R.$ Hence, the $\delta$-thinness of the geodesic triangles $\isom{r}{u} \cup \isom{u}{x} \cup P$ and $\sigma_i \cup \isom{v_i}{x} \cup P$ implies $\dist{u}{m_i}\le \dist{u}{m''_i}+\dist{m''_i}{m'_i}+\dist{m'_i}{m_i}\le R+2\delta.$
       \item {\bf Case 2.}  Now, assume that $(v_i|x)_r \le (u|x)_r.$
       Let $m_i,$ $m'_i,$ $m''_i$ be the vertices of $\sigma_i,$ $P$ and $\isom{r}{u}$ at distance $(v_i|x)_r$ from $r.$ By the choice of $v_i,$ $\dist{r}{m_i}+\dist{m_i}{v_i} = \dist{r}{v_i} \ge \dist{r}{u} = \dist{r}{m''_i}+\dist{m''_i}{u}.$ Since $\dist{r}{m_i} = \dist{r}{m''_i},$ we deduce that $\dist{m''_i}{u}\le \dist{m_i}{v_i}.$ Since $\dist{v_i}{\isom{r}{x}}\le R,$ \Cref{GP-Bound} implies $\dist{m''_i}{u} \le \dist{m_i}{v_i} = (r|x)_{v_i} \le R.$ Using the thinness of geodesic triangles $\isom{r}{u} \cup \isom{u}{x} \cup P$ and $\sigma_i \cup \isom{v_i}{x} \cup P,$ we derive that $\dist{u}{m_i}\le \dist{u}{m''_i}+\dist{m''_i}{m'_i}+\dist{m'_i}{m_i}\le R+2\delta.$ 
   \end{itemize}
       We conclude that any vertex $u\in V(G)$ is at distance at most $R+2\delta$ from some path in~$\mathcal{C}$.
       %\todo[inline]{Draw the picture.}
   \end{proof}
   
   \begin{figure}[t]
       \centering
       \begin{tabular}{cc}
       \includegraphics[width=0.47\linewidth]{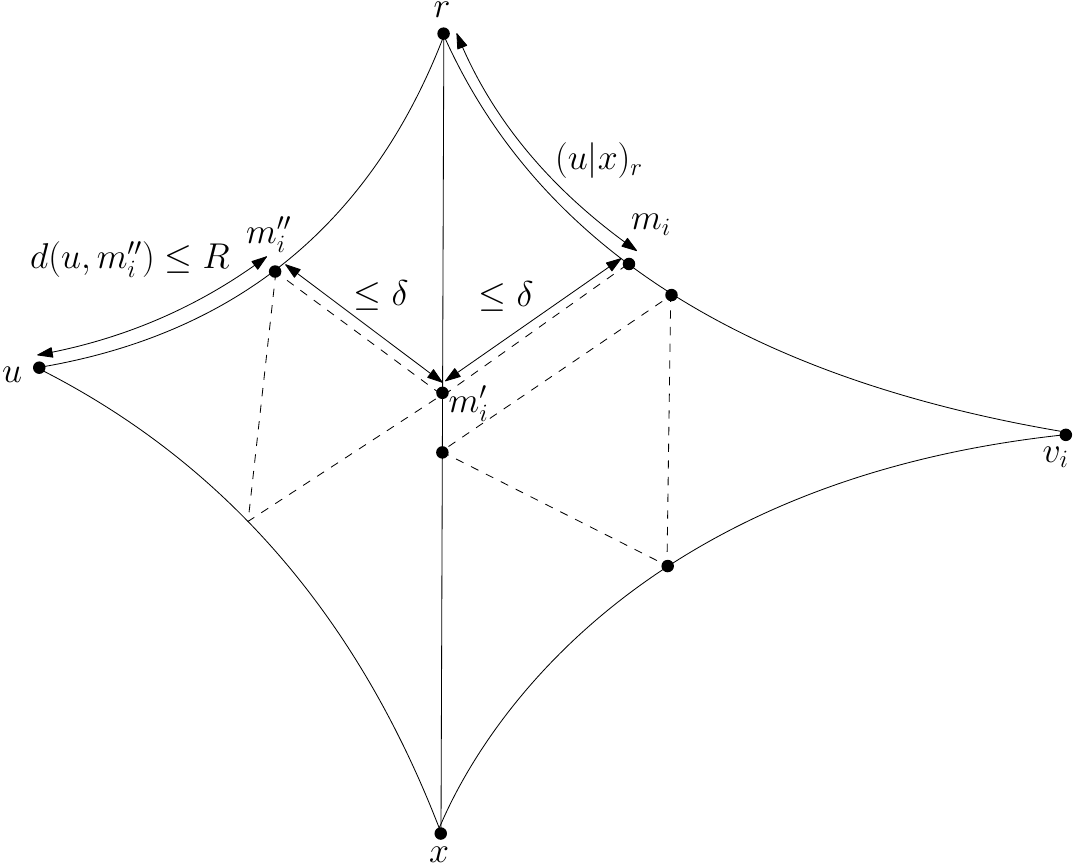} &  
       \includegraphics[width=0.47\linewidth]{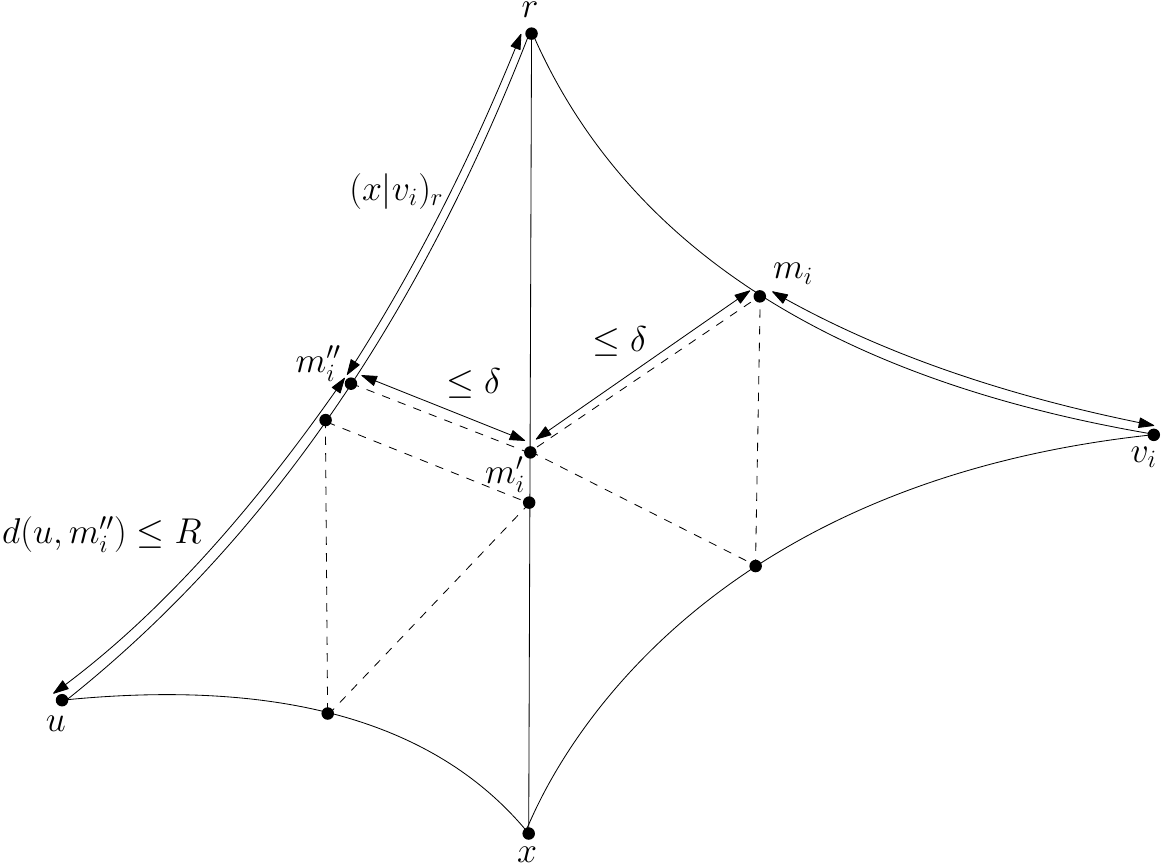}  \\
       {\bf Case 1.} $(u|x)_r \le (v_i|x)_r$ & {\bf Case 2.} $(u|x)_r > (v_i|x)_r$
       \end{tabular}
       \caption{Illustrations for the two cases in the proof of \Cref{lem:rooted-paths}.}
       \label{fig:Lemma-rooted-paths}
   \end{figure}
   
    \begin{algorithm}[t]
    \DontPrintSemicolon
   \begin{small}
   \SetKwInOut{KwIn}{Input}
       \SetKwInOut{KwOut}{Output}
       \KwIn{A $\delta$-thin graph $G$, an integer $k\leq |V(G)|$.}
       \KwOut{ An integer $R$ and an $(u,R+2\delta)$-cover of size $2k-1$ such that there is no $(v,R')$-cover with $R'<R.$}
   
       \For{$v\in V(G)$}{

           Let $R_v$ be the smallest $R$ integer for which Algorithm~\ref{algo:rooted-paths} returns a $(v,R+2\delta)$-cover $\mathcal{C}_v$ of size $2k-1.$

           \tcc{The above steps can be implemented by using Algorithm~\ref{algo:rooted-paths} in combination with a binary search on $R_v \in \{0,1, \ldots, |V(G)|\}$.}
       }   
       Let $u\in V(G)$ such that $R_u=\min\{R_v\colon v\in V(G)\}$.

       \Return $(\mathcal{C}_u,R_u)$.

%       Let $\pi$ be a profile of size $2k$ consisting of $u$ and $2k-1$ other end-vertices of the paths in $\mathcal{C}_u.$

%       Computes a $2\tau(G)+1$-shallow pairing $P$ of $\pi.$
       
%       \Return $\mathcal{S}(P)$.

      \caption{}\label{algo:guaranteed-rooted-path}
       
   \end{small}
    \end{algorithm}

    \begin{lemma}
       \label{lem:correct-rooted-cover}
       %Let $G$ be a $\delta$-thin graph, and $k$ be an integer. Then, 
       Algorithm~\ref{algo:guaranteed-rooted-path} returns an integer $R$ and a $(u,R+2\delta)$-cover $\mathcal{C}_u$ of $G$ of size $2k-1$ such that, there is no $(v,R')$-cover of size $2k-1$ with $R'< R.$ 
   \end{lemma}
   \begin{proof}
    %Let $G$ be a connected graph and $\delta=\tau(G)$ its thinness. 
    %Then, due to Proposition~\ref{thm:hyp-thin} the thinness $\delta$ of $G$ is at most $4\delta$. 
    Let $\left(\mathcal{C}_u,R_u\right)$ be the output of Algorithm~\ref{algo:guaranteed-rooted-path}. Since $R_u=\min\{R_v\colon v\in V(G)\}$ is the minimum integer $R$ for which Algorithm~\ref{algo:rooted-paths} returns $(u,R+2\delta)$-cover $\mathcal{C}_u,$ Algorithm~\ref{algo:rooted-paths} returns a $(v,R')$-packing of size $2k$ for any $R'<R_u$ and any $v\in V(G).$ Hence, there is no $(v,R')$-cover of size $2k-1$ with $R'< R.$ 
   \end{proof}

\subsection{From non rooted to rooted collection of paths}\label{sec:non-rooted}

In the following lemma, we show that if there is an $R$-cover of size $k$ of a $\delta$-thin graph then there is a rooted $(r,R+\delta)$-cover of size $2k-1$ of $G,$ for some $r\in V(G)$.

\begin{lemma}\label{lem:path-tree}
    Let $\mathcal{C}$ be an $R$-cover of a $\delta$-thin graph $G$ with $|\mathcal{C}|=k,$ $X_{\mathcal{C}}$ the set of end-vertices of paths in $\mathcal{C}$ and $r\in X_{\mathcal{C}}.$ Then, any collection isometric paths $\mathcal{C}_r = \left\{ \isom{r}{x} \colon x\in X_{C}\setminus \{r\} \right\}$ is an $(r, R+\delta)$-cover of $G$.
\end{lemma}
\begin{proof}

For a vertex $u\in V(G)$, let $P=\isom{v_1}{v_2} \in \mathcal{C}$ be a path such that $\dist{u}{P} \leq R$ and $w\in V(P)$ be a vertex with $\dist{u}{w} \leq R$. Let $P_1,P_2\in \mathcal{C}_r$ where $P_i = \isom{r}{v_i}$. Since the geodesic triangle $P\cup P_1 \cup P_2$ of $G$ is $\delta$-thin, either $\dist{w}{P_1}\leq \delta$ or $\dist{w}{P_2}\leq \delta$.  Therefore, either $\dist{u}{P_1} \leq R+\delta$ or $\dist{u}{P_2} \leq R+\delta$. Hence, $\mathcal{C}_r$ is a $(r, R+\delta)$-cover of $G$.
\end{proof}

\subsection{From rooted to non rooted collection of paths}
\label{sec:shallow-pairing-use}
\begin{figure}
    \centering
    \includegraphics*[width=0.5\textwidth]{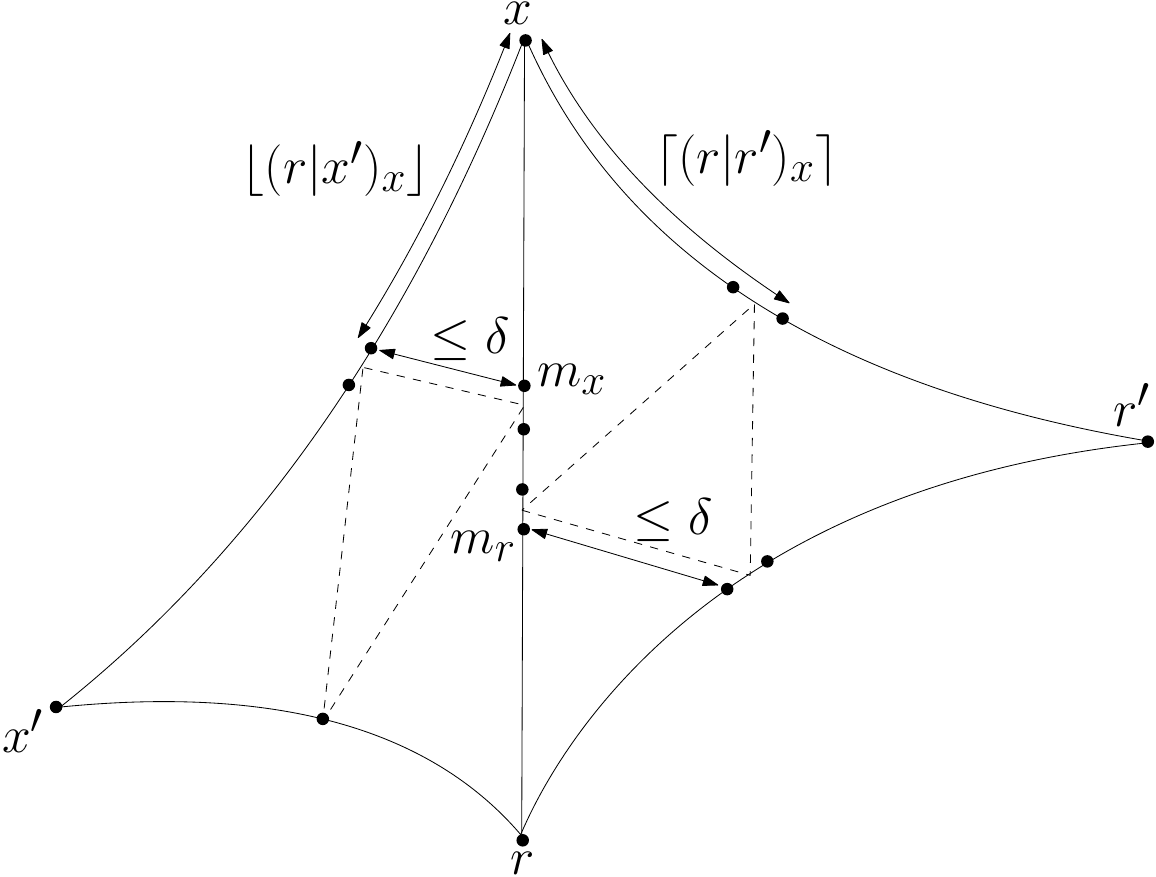}
%    \begin{tikzpicture}
%         \foreach \x/\y [count = \n] in {0/0, 0/3, 2/0.5, -2/0.5, 0/1, -1/0.25, -1/1.75, 0/2, 1/0.25, 1/1.75}
%    {
%    	\filldraw (\x, \y) circle (1.5pt); % draw the points 
%    }

%    \foreach \x/\y/\w/\z [count = \n] in {0/0/0/3, 0/3/2/0.5, 2/0.5/0/0, -2/0.5/0/0, -2/0.5/0/3}
%    {
%    	\draw (\x, \y) -- (\w,\z); % draw the points 
%    }

%    \draw[densely dotted] (0,1) -- (-1,0.25) -- (-1,1.75) -- cycle;

%    \draw[densely dotted] (0,2) -- (1,0.25) -- (1,1.75) -- cycle;
%    \node[above] at (0,3) {\scriptsize $x$}; \node[below] at (0,0) {\scriptsize $r$}; \node[left] at (-2,0.5) {\scriptsize $x'$}; \node[right] at (2,0.5) {\scriptsize $r'$}; \node[left] at (0,2) {\scriptsize $m_x$}; \node[right] at (0,1) {\scriptsize $m_r$};
    
%    \end{tikzpicture}
    \caption{Illustration of the notations used in Proof of \Cref{lem:shallow-pairing-bound}.}
    \label{fig:my_label}
\end{figure}

Conversely, the next lemma shows that, from a set of rooted isometric paths of a $\delta$-thin graph of size $2k-1$ and eccentricity $R$, it is possible to construct a $(R+3\delta+1)$-cover of $G$ of size $k$.

\begin{lemma}\label{lem:shallow-pairing-bound}
    Let $r$ be a vertex of a $\delta$-thin graph $G$. For integers $R$ and $k$, let $\mathcal{C}_r$ be an $(r,R)$-cover of $G$ with $|\mathcal{C}_r|=2k-1.$ Let $\pi_r$ be a profile of length $2k$ containing all end vertices of the paths in $\mathcal{C}_r$ and $P_r$ be a $(2\delta+\frac{1}{2})$-shallow pairing of $\pi_r$. Then, $\SP{P_r}$ is a $(R+3\delta+1)$-cover of $G$ of size $k.$ 
\end{lemma}

\begin{proof}
Let $u\in V(G)$ and $P=\isom{r}{x} \in \mathcal{C}_r $ be an $r$-path with $\dist{u}{P} \leq R$. Let $\{x',r'\}\subseteq \pi_r$ be the vertices such that $\{x,x'\}\in P_r$ and $\{r,r'\}\in P_r$. 
%In other words, $x'$ and $r'$ are vertices which are paired respectively with $x$ and $r$ in the $(2\tau(G)+\frac{1}{2})$-shallow pairing $P_r$ of $\pi_r$. 
By definition of a $(2\tau(G)+\frac{1}{2})$-shallow pairing, we have that $\grom{x}{x'}{m} \le 2\delta+\frac{1}{2}$  and $\grom{r}{r'}{m} \le 2\delta+\frac{1}{2}$ which imply the following inequalities: 
$$ \dist{x}{m} + \dist{m}{x'} - \dist{x}{x'} \leq 4\delta + 1 \Rightarrow \dist{x}{x'} \geq \dist{x}{m} + \dist{m}{x'} - (4\delta+1)$$ $$ \dist{r}{m} + \dist{m}{r'} - \dist{r}{r'} \leq 4\delta+1 \Rightarrow \dist{r}{r'} \geq \dist{r}{m} + \dist{m}{r'} - (4\delta+1) $$
Combining the above inequalities we derive:
%\begin{multline*}
%\dist{x}{x'} + \dist{r}{r'} \geq (\dist{x}{m} + \dist{m}{r}) + %(\dist{x'}{m} + \dist{m}{r'}) - (8\delta+2) \geq \dist{x}{r} + \dist{x'}{r'}\\  - (8\delta+2)
%\end{multline*}
\begin{multline}\label{eq:2}
\dist{x}{x'} + \dist{r}{r'} \geq (\dist{x}{m} + \dist{m}{r'}) + (\dist{x'}{m} + \dist{m}{r}) - (8\delta+2) \geq \dist{x}{r'}\\ + \dist{x'}{r}  - (8\delta+2)
\end{multline}
Adding $\dist{r}{x}$ to both sides of (\ref{eq:2}) we get: $$ \dist{r}{x} + \dist{x}{x'} - \dist{r}{x'} \geq \dist{r}{x} + \dist{x}{r'} - \dist{r}{r'} - (8\delta+2) $$ which further implies: 
\begin{equation}\label{eq:3}
    \grom{r}{x'}{x} \geq \grom{r}{r'}{x} - (4\delta+1)
\end{equation} 

Recall that $P=\isom{r}{x}\in \mathcal{C}_r$. Let $m_x$ be the point of the geodesic $P$ such that $\dist{x}{m_x} = \lfloor\grom{r}{x'}{x}\rfloor$  and $m_r$ be the point of the geodesic $P$ such that $\dist{x}{m_r} = \lceil \grom{r}{r'}{x} \rceil$.  Let $z$ be the vertex of $P$ with $\dist{u}{z} \leq R$. Consider the following cases.

\begin{itemize}
    \item If $z$ lies in the $(x,m_x)$-subpath of $P$, consider any isometric path $\isom{r}{x'}.$ Since the geodesic triangle $P\cup \isom{r}{x'}\cup \isom{x}{x'}$ is $\delta$-thin, we have that $\dist{z}{\isom{x}{x'}} \leq \delta$. Hence, $\dist{u}{\isom{x}{x'}} \leq \dist{u}{z} + \dist{z}{\isom{x}{x'}} \le R + \delta$.

    \item If $z$ lies in the $(r,m_r)$-subpath of $P$, 
    consider any isometric path $\isom{r'}{x}.$ Since the geodesic triangle $P\cup \isom{r}{r'}\cup \isom{r'}{x}$ is $\delta$-thin, we have that 
    $\dist{z}{\isom{r}{r'}} \leq \delta.$ Hence, $\dist{u}{\isom{r}{r'}} \leq \dist{u}{z} + \dist{z}{\isom{r}{r'}} \leq R + \delta$.

    \item Otherwise, $z$ must lie in the $(m_r,m_x)$-subpath of $P$. Due to inequality (\ref{eq:3}) we have that  $\dist{m_r}{m_x}\le 4\delta+2.$  This implies either $\dist{z}{m_x} \leq 2\delta+1$ or 
    $\dist{z}{m_r} \leq 2\delta+1$. Since the geodesic triangles $P\cup \isom{r}{x'}\cup \isom{x}{x'}$ and $P\cup \isom{r}{r'}\cup \isom{x}{r'}$ are
    $\delta$-thin, we have $\dist{m_x}{\isom{x}{x'}}\le \delta$ and $\dist{m_r}{\isom{r}{r'}}\le \delta.$ If $\dist{z}{m_x} \leq 2\delta+1$ then
    $\dist{u}{\isom{x}{x'}} \le \dist{u}{z}+\dist{z}{m_x}+\dist{m_x}{\isom{x}{x'}}\le  R+3\delta+1$.  Otherwise, $\dist{z}{m_r} \leq 2\delta+1$ and $\dist{u}{\isom{r}{r'}} \le \dist{u}{z}+\dist{z}{m_r}+\dist{m_r}{\isom{r}{r'}} \le R+3\delta+1$. 
\end{itemize}
In all three above cases, either
$\dist{u}{\isom{x}{x'}}\le R+3\delta+1$ or $\dist{u}{\isom{r}{r'}}\le R+3\delta+1.$ We conclude that, for any $u\in V(G)$,  the collection $\SP{P_r}$ contains an isometric path $Q$ such that $\dist{u}{Q}\le R+3\delta+1,$ i.e. $\SP{P_r}$ is a $(R+3\delta+1)$-cover. 
\end{proof}

\subsection{Proof of \Cref{thm:thinness}}\label{sec:complete}

    \begin{algorithm}[t]
    \DontPrintSemicolon
   \begin{small}
   \SetKwInOut{KwIn}{Input}
       \SetKwInOut{KwOut}{Output}
       \KwIn{A $\delta$-thin graph $G$, an integer $k\leq |V(G)|$.}
       \KwOut{ A collection of $k$ isometric paths with eccentricity at most $R^*_k+(6\tau(G)+1)$ }

       Let $\mathcal{C}_u$ be the $(u,R_u)$-cover of $G$ returned by Algorithm \ref{algo:guaranteed-rooted-path} with $G$ and $k$ as input.
    
       Let $\pi$ be the even profile consisting of $u$ and every other end vertex of the $u$-paths in $\mathcal{C}_u$.

    Computes a $(2\tau(G)+\frac{1}{2})$-shallow pairing $P$ of $\pi.$
       
      \Return $\mathcal{S}(P)$.
   
      \caption{}\label{algo:k-geodesic center}
       
   \end{small}
    \end{algorithm}

   Let $G$ be a connected graph and $\delta=\tau(G)$ its thinness. 
   %Then, due to Proposition~\ref{thm:hyp-thin} the thinness $\delta$ of $G$ is at most $4\delta$. 
   Let $\mathcal{C}^*$ be a $\opt{k}$-cover of $G,$ $\pi_{\mathcal{C}^*}$ the profile containing the end-vertices of paths in $\mathcal{C}^*$ and $r\in \pi_{\mathcal{C}^*}$. Consider the set $\mathcal{C}^*_r = \left\{ \isom{r}{x} \colon x\in \pi_{\mathcal{C}^*} \setminus \{r\} \right\}$. Due to \Cref{lem:path-tree}, we have
   \begin{equation}\label{eq1}
   \mathcal{C}^*_r \text{ is a } \left(r,\opt{k}+\delta\right)\text{-cover of } G.    
   \end{equation}
    Now let $\left(\mathcal{C}_u,R_u\right)$ be the output of Algorithm~\ref{algo:guaranteed-rooted-path}. By Lemma \ref{lem:correct-rooted-cover}, there is no $(r,R')$-cover of size $2k-1$ with $R'<R_u.$ Due to (\ref{eq1}), there exists a $(r,\opt{k}+\delta)$-cover of $G$ of size $2k-1.$ Hence, $R_u \le \opt{k} + \delta$ and $\mathcal{C}_u$ has eccentricity at most $R_u+2\delta\le \opt{k}+3\delta.$ Let $\pi_u$ be the profile consisting of $u$ and every other end vertex of the $u$-paths in $\mathcal{C}_u$. Observe that $\pi_u$ is an even profile of length of $2k.$ By \Cref{lem:algo-shallow-pairing}, it is possible to compute in $O(mn^2)$-time a $(2\tau(G)+\frac{1}{2})$-shallow pairing $P'$ of $\pi_u$. Due to \Cref{lem:shallow-pairing-bound}, $\SP{P'}$ is an $(\opt{k} + 6\delta+1)$-cover of size $k.$ 
    %Hence, $\mathcal{C}'$ is an $(\opt{k} + 6\delta+1)$-cover of $G$ (where $\delta$ is the thinness of $G$) and an $(\opt{k} + 24\delta+1)$-cover of $G$ (where $\delta$ is the hyperbolicity of $G$). 
    This completes the proof of \Cref{thm:thinness}. \Cref{algo:k-geodesic center} describes our complete algorithm for $k$-\textsc{Geodesic Center} on $\delta$-thin graphs. %\todo{Added a line here. Is it OK? Yes I think it's ok. Should we include the running time and/or refer to Algorithm 3 in the statement of Theorem 10 ? Ans: I agree with the changes in the statement of Theorem 10.}

    Clearly, the running times of \Cref{algo:rooted-paths} and \Cref{algo:guaranteed-rooted-path} are $O(k(n+m))$ and $O(nk(n+m)\log n)$, respectively. Due to \Cref{lem:algo-shallow-pairing}, computing a shallow pairing takes $O(mn^2)$-time. 
    %(We assume that the thinness as part of the input). 
    Therefore, the total running time of Algorithm \ref{algo:k-geodesic center} is $O(mn^2\log n)$. 

    \subsection{The special case of trees}

    In case of trees, the same algorithmic approach leads to an exact polynomial time algorithm. Indeed, since trees are $0$-hyperbolic, Lemma \ref{lem:correct-rooted-cover} implies that Algorithm \ref{algo:guaranteed-rooted-path} computes a rooted $(r,R)$-cover $\mathcal{C}$ of size $2k-1$ such that there is no $(r',R')$-cover with $R'<R.$ By Lemma \ref{lem:rooted-paths}, an optimal $R^*$-cover of size $k$ can be transformed into a rooted $(r',R^*)$-cover of size $2k-1.$ Hence, $R^*\ge R.$ Since trees satisfy the pairing property, any $(r,R)$-cover of size $2k-1$ can be transformed into an $R$-cover of size $k.$ This implies $R \ge R^*.$ Hence, in case of trees, $R=R^*$ and the solution returned by Algorithm \ref{algo:k-geodesic center} is an optimal $R^*$-cover of size $k$.
    
    % \todo{Do you have any suggestion on where we can submit this paper? Is it ok for you to submit directly to a journal or you prefer to submit first to a conference ? In the first case, we could try to submit to a journal in which [22] or [23] appeared (i.e. JGAA or TCS), what do you think ? Another alternative is Networks where [27] appeared.}
 %\todo{I propose to add Algorithm 3. What do you think about that ? If you agree then we have to check in the text where we should refer to it. Ans: Good idea to include Algorithm 3.}

\newcommand{\specialset}[4]{S^{#1}_{#2}\left( #3, #4\right)}

% Let $X^*$ be the set of end vertices of the paths of $\mathcal{C}^*_r$. Observe that $|X^*|=2k$, and by \Cref{lem:pairing-algo} it is possible to compute a shallow pairing $\mathcal{S}^*$ of $X^*$. Due to \Cref{pairing-approx}, $\mathcal{S}^*$ is a $(R+3\delta+1)$-cover of $G$. Hence, 

\newcommand{\variableVertex}[3]{v^{#1}_{\left(#2,#3\right)}}

\newcommand{\subcubic}[1]{B\left(#1\right)}
\newcommand{\subcubicp}[1]{B'_k\left(#1\right)}
\newcommand{\partialG}[1]{P_k\left(#1\right)}

\section{NP-hardness for partial grids}\label{sec:hard}

In this section we prove \cref{thm:hard}. Our proof is an adaptation of the NP-hardness of $1$-\textsc{Geodesic Center} on planar bipartite graphs proved by Dragan \& Leitert (Corollary 8, \cite{dragan2017minimum}). First we prove the following lemmas.

\begin{lemma}\label{lem:cov-subdivide}
    Let $G$ be a graph and $H=G_{\ell}$ for some $\ell\geq 1$. Then for integers $m,k$, if $G$ has an $m$-cover of size $k$ then $H$ has an $\left( m\ell+\left\lfloor \ell/2\right\rfloor\right)$-cover of size $k$.
    % \todo{Is it true that the lemma holds with $\left\lfloor \ell/2\right\rfloor$ instead of $\left\lceil \ell/2\right\rceil$ ? Indeed, any vertex of $H$ is at distance at $\left\lfloor \ell/2\right\rfloor$ from some original vertex.}
\end{lemma}
\begin{proof}
    Let $\mathcal{C}$ be an $m$-cover of $G$ of size $k$ and $\mathcal{C}'$ be a set of paths in $H$ which are $\ell$-subdivision of the paths in $\mathcal{C}$. Clearly, all paths in $\mathcal{C}'$ consists of isometric paths in $H$. Let $u$ be an original vertex of $H$ and $P\in \mathcal{C}$ be a path such that $\dist{u}{V(P)}\leq m$ in $G$. Let $P'\in \mathcal{C}'$ be the $\ell$-subdivision of $P$. Clearly, $\dist{u}{V(P')}\leq m\ell$. Now consider a vertex $u\in V(H)$ which is not a vertex of $G$. Hence there exists an original vertex $u'\in V(H)$ with $\dist{u}{u'}\leq \left\lfloor \ell/2\right\rfloor$. Let $P\in \mathcal{C}$ be a path such that $\dist{u'}{V(P)}\leq m$ in $G$ and $P'\in \mathcal{C}'$ be the $\ell$-subdivision of $P$.  Then $\dist{u}{V(P')} \leq \dist{u}{u'} + \dist{u'}{V(P')} \leq m\ell+ \left\lfloor \ell/2\right\rfloor$. 
\end{proof}
 Let $G$ be a graph and $H=G_{\ell}$ for some $\ell\geq 1$.  For an isometric path $P$ of $H$ between two original vertices $u,v$, let $G(P)$ denote the $(u,v)$-isometric path in $G$ such that $P$ is an $\ell$-subdivision of $G(P)$. Intuitively, $G(P)$ is the original isometric path whose $\ell$-subdivision created $P$ in $H$. 
 
\begin{lemma}\label{lem:subdivide-original}
    Let $G$ be a graph and let $H=G_{\ell}$ for some $\ell\geq 1$. Let $P$ be an isometric path of $H$ between two original vertices $u,v$. Let $w$ be an original vertex of $G$ such that $\dist{w}{V(P)} < (r+1)\ell$ for some positive integer $r$. Then, for $Q=G(P)$ we have $\dist{w}{V(Q)}\leq r$. 
    % \todo{Since $w$ is an original vertex its distance to $V(P)$ should be a multiple of $\ell$ and we can suppose that $\dist{w}{V(P)} < (r+1)\ell$ instead of $\dist{w}{V(P)} \leq r\ell+\left\lceil \ell/2\right\rceil$}
\end{lemma}

\begin{proof}
    Let $w'\in V(P)$ be a vertex which is closest to $w$ in $H$ and $P'$ be an $(w,w')$-isometric path in $H$. Clearly,  $w'$ is an original vertex and therefore $G(P')$ exists. Observe that, the number of original vertices in $P'$ is at most $r+1$. (Otherwise length of $P'$ is at least $(r+1)\ell$ in $H$ which is a contradiction.) Hence length of $G(P)$ is at most $r$ in $G$. 
\end{proof}

  Dragan \& Leitert \cite{dragan2017minimum} reduced the NP-complete \textsc{Planar Monotone 3-SAT}~\cite{de2010optimal} to show that $1$-\textsc{Geodesic Center} is NP-hard on bipartite planar subcubic graphs. Given an \textsc{Planar Monotone 3-SAT} instance $I$, the authors constructed a planar bipartite subcubic graph $\subcubic{I}$ and an integer $m'$ with the following properties. 

 \begin{itemize}
     \item $\subcubic{I}$ has an isometric path with eccentricity at most $m'$ if and only if $I$ is satisfiable;

     \item there are two special cut vertices $v_0,v_n$ of $\subcubic{I}$ such that any isometric path with eccentricity at most $m'$ will contain $v_0$ and $v_n$.
 \end{itemize}
  
To prove our result, we modify the graph $\subcubic{I}$ slightly.  First construct a gadget as follows. Take a path $P$ of length $2k$ and let the vertices of $P$ be $u_1, u_2, \ldots, u_{2k+1}$. For each $j\in [2,2k]$, take a new path $Q_j$ of length $m'$ and make one of the end-vertex of $Q_j$ adjacent to $u_j$. Let $T$ be the union of $P$ and $Q_j, j\in [2,2k-1]$. Now make the vertex $v_0$ adjacent to $u_1$. %\todo[]{Since only vertices $u_j,$ $j\in [2,2k-1]$ are connected to a path $Q_j.$ It seems to me that the length of the path should be $2k$ instead of $2k+1$ ?}

We call the modified graph $\subcubicp{I}$. It is easy to verify that a set $\mathcal{P}$ of isometric paths in $\subcubicp{I}$ is an $m'$-cover if and only if the following holds:

\begin{itemize}
    \item There are $k-1$ isometric paths in $\subcubicp{I}$ whose vertices completely lie in $T$.

    \item There is a special isometric path $P$ in $\subcubicp{I}$ containing $v_0,v_n$ such that $P$ has eccentricity $m'$ in $\subcubic{I}$.
\end{itemize}

The above discussion implies that $\subcubicp{I}$ has an $m'$-cover of size $k$ if and only if $\subcubic{I}$ has an $m'$-cover of cardinality one. Moreover, $\subcubicp{I}$ is planar and has maximum degree at most $3$. Now we construct a $H=\partialG{I}$ by applying \Cref{lem:equal} on $G=\subcubicp{I}$ and $\ell$ be the integer such that $H=G_{\ell}$. We prove that $G$ has an $m'$-cover of cardinality $k$ if and only if $H$ has an $\left( m'\ell+\left\lceil \ell/2\right\rceil\right)$-cover of cardinality $k$.

    If $G$ has an $m'$-cover of cardinality $k$, then \Cref{lem:cov-subdivide} implies that $H$ has an $\left( m'\ell+\left\lceil \ell/2\right\rceil\right)$-cover of cardinality $k$. Now assume that $H$ has an $\left( m'\ell+\left\lceil \ell/2\right\rceil\right)$-cover $\mathcal{C}$ of cardinality $k$. Let $T_{\ell}$ be the induced subgraph of $\partialG{I}$ isomorphic to the $\ell$-subdivision of $T$. The structure of $T_{\ell}$ implies that there will be $k-1$ isometric paths whose vertices lies completely in $T_{\ell}$. Let $B_{\ell}$ denote the subgraph of $\partialG{I}$ isomorphic to the $\ell$-subdivision of $\subcubic{I}$. Since the original vertices $v_0$ and $v_n$ are still cut vertices in $\partialG{I}$, there is a special path $P\in \mathcal{C}$ such that all vertices in $B_{\ell}$ is at a distance at most $m'\ell+\left\lceil \ell/2\right\rceil$ from $P$. Now apply \Cref{lem:subdivide-original} to conclude that all vertices of $\subcubic{I}$ which also belongs to $\subcubicp{I}$ is at distance $m'$ from $G(P)$. Therefore, the set $\mathcal{C}'=\{G(P)\colon P\in \mathcal{C}\}$ is an $m'$-cover of $G$.

The above discussion implies that $\partialG{I}$ has an $\left( m'\ell+\left\lceil \ell/2\right\rceil\right)$-cover of cardinality $k$ if and only if $I$ is satisfiable. This completes the proof.

\section{Conclusion}

Constant factor approximability of $k$-\textsc{Geodesic Center} in general graphs for $k\geq 2$ remains open. From both application and theoretical point of views, approximability of $k$-\textsc{Geodesic Center} on planar graphs is an important open question. An approximation scheme even for \textsc{MESP} on planar graphs is open. Studying whether $k$-\textsc{Geodesic Center} (for arbitrary $k$) is fixed-parameter tractable with respect to tree-width or hyperbolicity are interesting research questions. Polynomial solvability of $k$-\textsc{Geodesic Center} in superclasses of trees (e.g. outerplanar graphs), solid grids (i.e. partial grids where internal faces have unit area), are  open as well. 
For graphs with bounded isometric path complexity (including hyperbolic graphs), \textsc{IPC} admits a constant factor approximation algorithm \cite{ChakrabortyCFV23}. It would be interesting to explore if $k$-\textsc{Geodesic Center} admits an additive approximation algorithm on graphs with bounded isometric path complexity? Finally, investigating the approximability or fixed-parameter tractability of $k$-\textsc{Geodesic Center} on weighted graphs are interesting directions.

\bibliographystyle{plain}

\bibliography{references}

\end{document}

%% file: delta-thin.pdf_t
\begin{picture}(0,0)%
\includegraphics{delta-thin.pdf}%
\end{picture}%
\setlength{\unitlength}{2279sp}%
\begingroup\makeatletter\ifx\SetFigFont\undefined%
\gdef\SetFigFont#1#2#3#4#5{%
  \reset@font\fontsize{#1}{#2pt}%
  \fontfamily{#3}\fontseries{#4}\fontshape{#5}%
  \selectfont}%
\fi\endgroup%
\begin{picture}(7921,3389)(2863,-5690)
\put(8956,-3616){\makebox(0,0)[rb]{\smash{{\SetFigFont{8}{9.6}{\rmdefault}{\mddefault}{\updefault}{\color[rgb]{0,0,0}$\alpha_y$}%
}}}}
\put(4016,-4831){\makebox(0,0)[b]{\smash{{\SetFigFont{8}{9.6}{\rmdefault}{\mddefault}{\updefault}{\color[rgb]{0,0,0}$\le \delta$}%
}}}}
\put(5111,-4821){\makebox(0,0)[b]{\smash{{\SetFigFont{8}{9.6}{\rmdefault}{\mddefault}{\updefault}{\color[rgb]{0,0,0}$\le \delta$}%
}}}}
\put(9286,-4403){\makebox(0,0)[b]{\smash{{\SetFigFont{8}{9.6}{\rmdefault}{\mddefault}{\updefault}{\color[rgb]{0,0,0}$m'$}%
}}}}
\put(6683,-3383){\makebox(0,0)[b]{\smash{{\SetFigFont{8}{9.6}{\rmdefault}{\mddefault}{\updefault}{\color[rgb]{0,0,0}$\varphi$}%
}}}}
\put(5349,-4111){\makebox(0,0)[b]{\smash{{\SetFigFont{8}{9.6}{\rmdefault}{\mddefault}{\updefault}{\color[rgb]{0,0,0}$m_x$}%
}}}}
\put(3774,-4104){\makebox(0,0)[b]{\smash{{\SetFigFont{8}{9.6}{\rmdefault}{\mddefault}{\updefault}{\color[rgb]{0,0,0}$m_z$}%
}}}}
\put(4584,-5191){\makebox(0,0)[b]{\smash{{\SetFigFont{8}{9.6}{\rmdefault}{\mddefault}{\updefault}{\color[rgb]{0,0,0}$m_y$}%
}}}}
\put(2878,-5598){\makebox(0,0)[b]{\smash{{\SetFigFont{8}{9.6}{\rmdefault}{\mddefault}{\updefault}{\color[rgb]{0,0,0}$x$}%
}}}}
\put(6274,-5617){\makebox(0,0)[b]{\smash{{\SetFigFont{8}{9.6}{\rmdefault}{\mddefault}{\updefault}{\color[rgb]{0,0,0}$z$}%
}}}}
\put(7349,-5617){\makebox(0,0)[b]{\smash{{\SetFigFont{8}{9.6}{\rmdefault}{\mddefault}{\updefault}{\color[rgb]{0,0,0}$x'$}%
}}}}
\put(10769,-5617){\makebox(0,0)[b]{\smash{{\SetFigFont{8}{9.6}{\rmdefault}{\mddefault}{\updefault}{\color[rgb]{0,0,0}$z'$}%
}}}}
\put(4570,-2511){\makebox(0,0)[b]{\smash{{\SetFigFont{8}{9.6}{\rmdefault}{\mddefault}{\updefault}{\color[rgb]{0,0,0}$y$}%
}}}}
\put(9047,-2508){\makebox(0,0)[b]{\smash{{\SetFigFont{8}{9.6}{\rmdefault}{\mddefault}{\updefault}{\color[rgb]{0,0,0}$y'$}%
}}}}
\put(4546,-3886){\makebox(0,0)[b]{\smash{{\SetFigFont{8}{9.6}{\rmdefault}{\mddefault}{\updefault}{\color[rgb]{0,0,0}$\le \delta$}%
}}}}
\put(8146,-4876){\makebox(0,0)[rb]{\smash{{\SetFigFont{8}{9.6}{\rmdefault}{\mddefault}{\updefault}{\color[rgb]{0,0,0}$\alpha_x$}%
}}}}
\put(10036,-4921){\makebox(0,0)[lb]{\smash{{\SetFigFont{8}{9.6}{\rmdefault}{\mddefault}{\updefault}{\color[rgb]{0,0,0}$\alpha_z$}%
}}}}
\end{picture}%